\documentclass[conference]{IEEEtran}
\ifCLASSINFOpdf
  \usepackage[pdftex]{graphicx}
  \graphicspath{{images/}}
  \DeclareGraphicsExtensions{.pdf,.jpeg,.png}
\else
\fi
\usepackage{amsmath}
\usepackage{array}
\usepackage{url}

\usepackage{amsthm}

\newtheorem{definition}{Definition}
\newtheorem{lemma}{Lemma}
\newtheorem{theorem}{Theorem}
\newtheorem{corollary}{Corollary}

\begin{document}
%
\title{Regular mixed-radix DFT matrix factorization for in-place FFT accelerators}

\author{\IEEEauthorblockN{Sergey Salishev}
\IEEEauthorblockA{Department of Computer Science\\
Saint Petersburg University\\
 University Embankment, 7/9, 199034, Russia\\
Email: s.salischev@spbu.ru}}


%


\maketitle

\begin{abstract}
The generic vector memory based accelerator is considered which supports DIT and DIF FFT with fixed datapath.
The regular mixed-radix factorization of the DFT matrix coherent with the accelerator architecture is proposed and the correction proof is presented.
It allows better understanding of architecture requirements and simplifies the developing and proving correctness of more complicated algorithms and
conflict-free addressing schemes.
\end{abstract}


%
\IEEEpeerreviewmaketitle

\section{Introduction}
\IEEEPARstart{F}{ast} Fourier Transform (FFT) processor efficiency is crucial for overall performance in many applications. 
In-place memory-based approach to FFT processor architecture is commonly used for many high-throughput communication tasks. 

The use of this approach guarantees that for each butterfly or group of butterflies both inputs and results are stored in the same memory locations, 
so for complex FFT sampled at N points a memory storing N complex words can be used. If faithfully implemented it provides additional 
benefit of runtime reconfiguration for different FFT sizes.

To maximize the throughput the memory bandwidth should be fully utilized. It is achieved if butterfly size is selected equal to number of parallel memory banks and one butterfly is initiated per clock.
To simplify the circuit design and reduce power and area it is also important to remove complex flow control and memory access arbitration.
To do so each wing of the butterfly should read and write non-conflicting memory ports. This requires conflict-free data to bank assignment. 

Johnson~\cite{johnson1992conflict}  suggested an in-place addressing strategy and architecture that allows launch of one butterfly per clock for pure-radix FFT, 
he also suggested modification of this scheme for mixed-radix launching one butterfly per clock. 
In-place FFT algorithms use bit/digit-reverse data permutations. If FFT is used for fast convolution intermediate data can be kept bit/digit-reversed. 
If FFT is a part of complex DSP algorithm data permutation substantially complicates its use. This can be resolved by so-called ``self-sorting'' in-place FFT which 
has both input and output data in natural order.

Jo and Sunwoo~\cite{jo2005new} suggested an in-place addressing strategy and architecture for radix 4/2 FFT launching 2 radix 2 butterflies in radix-2 stage simultaneously. 
Hsiao, Chen and Lee~\cite{hsiao2010generalized} suggested an in-place addressing strategy and architecture for arbitrary mix-radix FFT launching one butterfly
per clock. 
All the above algorithms use dual-port (1r1w) SRAM.
Salishev~\cite{salishev2014continuous} proposed an unified accelerator architecture and corresponding conflict-free addressing strategies for mixed-radix $R/r$ 
FFT of DIT, DIF, and self-sorting type and with some modifications can utilize single-port (1rw) SRAM. 
Later Salishev and Shein~\cite{salishev2013fft} presented proofs of correctness for these addressing strategies in rather technical manner by analysis of butterfly wing indices for possible conflicts. 

The paper revisits results of~\cite{salishev2013fft}.
The contribution of this paper is a regualr mixed-radix factorization of DFT matrix which is coherent with the accelerator architecture and 
supports FFT length and DIT/DIF selection in runtime with a fixed datapath (theorems~\ref{FFT_iter},~\ref{FFT_iter_freqW}). 
Each linear operator is mapped to specific stage of the datapath. The DIF/FIF and FFT length mode selection only requires runtime reconfiguration of the address generators.

\section{General approach to memory-based FFT accelerators} 

Memory-based continuous-flow FFT algorithms compute one butterfly per clock without blocking, and allow 
\begin{itemize}
\item variable in runtime FFT length;
\item usage of SRAM memory IP from standard library;
\item resource sharing between multiple algorithms.
\end{itemize}

The scheme of radix $R$ FFT IP block is shown on the Fig.~\ref{fig:fft-base}. This block can be considered as an anti-machine.
Processing Unit (PU) computes one butterfly per clock, using $R$ words from each of the memory banks. The PU computing resources can be used for 
vector operations of complex addition and multiplication, calculating elementary functions, etc.

\begin{figure}[htbp]
	\center
    \includegraphics[width=0.8\linewidth]{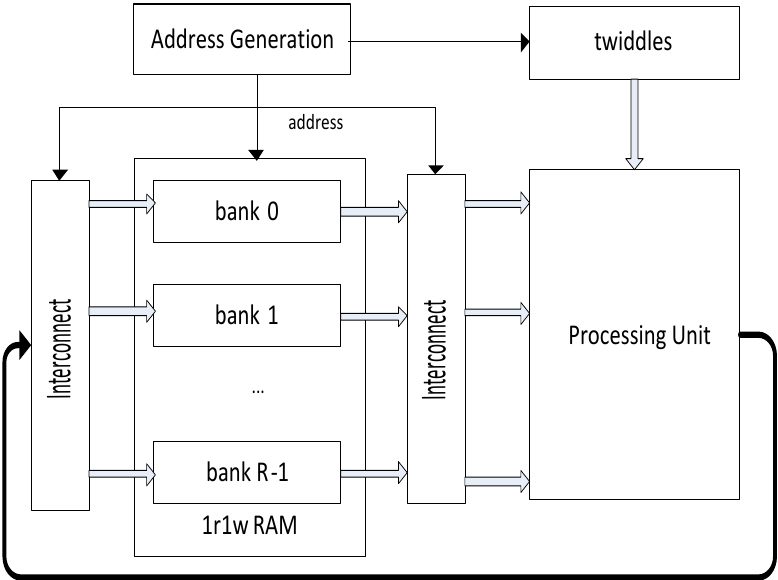}\\
    \caption{Architecture of FFT IP 1r1w RAM.}
    \label{fig:fft-base}
\end{figure}

Let us consider FFT of length $N=R^q$. 
If the transposition in the Interconnect block has 0 delay and reads/writes happen in the same clock (overlap), then 
the FFT computation time is 
$$
T(N) = \frac{N}{R} \log_R(N) + C_p,
$$
where $C_p$ --- is the PU pipeline length. 

We have set the following constraints for the algorithms in consideration:
\begin{itemize} 
\item compute time is $T(N)$
\item additional memory size is independent from $N$
\item using dual-port 1rw SRAM;
\item full memory bandwidth utilization; 
\item no dynamic memory access arbitration;
\item arbitrary butterfly radix;
\item variable in runtime FFT length; 
\end{itemize}

Let us consider FFT accelerator with 1r1w SRAM. 

\section{FFT splitting rule}

Let us $n\in\sf N$ and $\omega_n=e^{-\frac{2\pi i }{n}}$. DFT of length $n$ is a linear operator 
$$
{\mathcal F}_n=[\omega^{kl}_n]_{0\leq k,l<n}.
$$
For any $N>0$ and $0\le j<N$ denote by $e_{j,N}$ $j$--th unit vector in ${\sf R}^N$. 

Let us denote by $\otimes$ Kronecker product: $C=A\otimes B$, $A\in\sf R^{m\times n}$ and $B\in\sf R^{k\times \ell}B$, 
$C_{ik+j, p\ell+q} = A_{i,j} B_{p,q}$ for $0\le i< m$, $0\le j< n$, $0\le p< k$, $0\le q< \ell$. 

Let us $m, k \in \sf N$ and $n=km$. Set of $(e_{i,k}\otimes e_{j,m})$ for $0\le i<k$, $0\le j<m$ is a standard basis in ${\sf R}^n$. 

Let us define transposition matrix $L^n_k$ and exponent matrix $W^n_m$ of size $n$, $0\le i<k$, $0\le j<m$,
\begin{eqnarray*}
L^n_k (e_{i,k} \otimes e_{j,m}) & = & e_{j,m} \otimes e_{i,k}, \\
W^n_m (e_{i,k} \otimes e_{j,m}) & = & \omega_n^{ij} (e_{i,k} \otimes e_{j,m}).
\end{eqnarray*}

Transposition $L^n_k$ does the following index mapping:
$$
l^n_k : im+j \to jk+i.
$$
From the definition of $L^n_k$, for any $x\in {\sf R}^k$ and $y\in {\sf R}^m$: $$L^n_k (x\otimes y) = y \otimes x.$$
 
Matrix $W^n_m$ can be represented as follows:
$$
W^n_m = \left(
\begin{array}{lll}
	(V^m_n)^0& & \\
	&\ddots& \\
	& &(V^m_n)^{k-1} \\
\end{array}
\right),$$
$$
V^m_n=\left(
\begin{array}{lll}
	(\omega_n)^0& & \\
	&\ddots& \\
	& &(\omega_n)^{m-1} \\
\end{array}
\right).
$$

\begin{lemma} \label{aux_LW}
Let us $n=km$. Then
\begin{eqnarray*}
L^n_k & = & (L^n_m)^{-1} = (L^n_m)^T, \\
W^n_m & = & L^n_m W^n_k L^n_k.
\end{eqnarray*}
\end{lemma}

\begin{proof} 
From definitions.
\end{proof}

Kronecker product is non-commutative. Transposing of $I_k$ and ${\mathcal F}_m$ is performed using $L^n_k$.

\begin{lemma} \label{commutat}
Let's $n=km$. Then
$$
L^n_k (I_k \otimes {\mathcal F}_m) L^n_m = {\mathcal F}_m \otimes I_k.
$$
\end{lemma}

\begin{proof}
Let's $0\le i<k$ and $0\le j < m$. Then
\begin{eqnarray*}
L^n_k (I_k \otimes {\mathcal F}_m) L^n_m (e_{j,m}\otimes e_{i,k}) & = L^n_k (e_{i,k} \otimes ({\mathcal F}_m e_{j,m})) = \\
({\mathcal F}_m e_{j,m}) \otimes e_{i,k} & = ({\mathcal F}_m \otimes I_k)(e_{j,m}\otimes e_{i,k}).
\end{eqnarray*}
which proves the statement of the lemma.
\end{proof}

FFT algorithms are based on the factorization of the DFT matrix using splitting rule as follows.

\begin{lemma} \label{FFT_split}
Let's $n=km$, where $k,m \in \sf N$. Then
$$
{\mathcal F}_n = L^n_k (I_k \otimes {\mathcal F}_m) W^n_m ({\mathcal F}_k \otimes I_m).
$$
\end{lemma}

\begin{proof}
Let's $0\le i<k$, $0\le j < m$. Then 
\begin{eqnarray*}
W^n_m({\mathcal F}_k \otimes I_m) (e_{i,k} \otimes e_{j,m}) & = \\
W^n_m(({\mathcal F}_k e_{i,k}) \otimes e_{j,m}) & = \\
((\omega_k^{ip} \omega_n^{pj})_{p=0}^{k-1} \otimes e_{j,m}).
\end{eqnarray*}
As $(L^n_k)^T=L^n_m$, then for $0\le p<k$, $0\le q < m$, 
\begin{eqnarray*}
(e_{q,m}^T \otimes e_{p,k}^T) L^n_k (I_k \otimes {\mathcal F}_m) & = \\
\bigg((I_k \otimes {\mathcal F}_m^T)L^n_m (e_{q,m} \otimes e_{p,k})\bigg)^T & = \\
\bigg((e_{p,k} \otimes {\mathcal F}_m^T e_{q,m})\bigg)^T.
\end{eqnarray*}
Then
\begin{equation*}
\begin{split}
(e_{q,m}^T \otimes e_{p,k}^T) L^n_k (I_k \otimes {\mathcal F}_m)W^n_m({\mathcal F}_k \otimes I_m) (e_{i,k} \otimes e_{j,m}) = \\
\omega_k^{ip} \omega_n^{pj} (e_{q,m}^T{\mathcal F}_m e_{j,m}) = \omega_n^{ipm + pj + kjq}.
\end{split}
\end{equation*}
From the definition of $\omega_n$ it follows that $\omega_n^n=1$. Then
\begin{equation*}
\begin{split}
\omega_n^{ipm + pj + kjq} = \omega_n^{(p+qk)(im+j) - iqn} = \\
(e_{q,m}^T\otimes e_{p,k}^T) {\mathcal F}_n (e_{i,k} \otimes e_{j,m}),
\end{split}
\end{equation*}

which proves the statement of the lemma.
\end{proof}

\begin{corollary} \label{FFT_split_IF}
Let us $n=km$, and $k, m \in \sf N$. Then
$$
{\mathcal F}_n = L^n_k(I_k \otimes {\mathcal F}_m) W^n_m L^n_m (I_m \otimes {\mathcal F}_k) L^n_k.
$$
\end{corollary}

\begin{proof}
The result follows from \ref{FFT_split} by substitution with the formula from lemma \ref{commutat}.
\end{proof}

\section{Index inversion}

The general FFT formula results from iterating of splitting rule (corollary \ref{FFT_split_IF}). This formula includes multindex inversion described below.

\begin{definition}
Multi-index is a tuple $\alpha = (n_K, n_{K-1}, \ldots, n_0) \in \sf N^{K+1}$.
\end{definition}
For $N=\prod_{k=0}^K n_k$, any $n=0..N-1$ can be uniquely represented in the numeration system, generated by a multi-index $\alpha$, 
as a multi-index $p=(p_K, \ldots, p_0)$ using the condition
\begin{equation*}
\begin{split}
n = p_0 + n_0(p_1 + \cdots + n_{K-2}(p_{K-1} + n_{K-1}p_K) \cdots ), \\
0\le p_j < n_j, \quad 0\le j\le K.
\end{split}
\end{equation*}
In this case we will denote the relation between multi-index $p$ and number $n$ as: $p = n_{\alpha}$, $n = p^{\alpha}$. 
Multi-index $\alpha$ is generating a numbering system, multi-index $p$ --- matched with $\alpha$. 
Number of encoded numbers is denoted as $N=|\alpha|$.

From the definition it is clear that for any generating multi-indexes $\alpha$, $\beta$ and for any matched multi-indexes $p$, $q$, correspondingly,
$$
e_{p^{\alpha}, |\alpha|} \otimes e_{q^{\beta}, |\beta|} = e_{(p, q)^{(\alpha, \beta)}, |\alpha|\, |\beta|}.
$$

Multi-index inversion of  $p=(p_K, \ldots, p_0)$ is the reverse orderd multi-index: $p^{\star} = (p_0, \ldots, p_K)$. 
Multi-index $\alpha^{\star}=(n_0, \ldots, n_K)$ also generates the numbering system for numbers $0..N-1$. 
Transposition of multi-index inversion $\alpha$ is a transposition $P_{\alpha}$ of numbers $0..N-1$, definded by 
\begin{equation*}
\begin{split}
P\bigg(p_K + n_K(p_{K-1} + \cdots + n_2(p_1 + n_1 p_0) \cdots )\bigg) = \\ 
p_0 + n_0(p_1 + \cdots + n_{K-2}(p_{K-1} + n_{K-1}p_K) \cdots )
\end{split}
\end{equation*}
for $0\le p_j < n_j$, $0\le j\le K$. 
It can be rewritten as $P_{\alpha}n = ((n_{\alpha^{\star}})^\star)^{\alpha}$. 
So the transposition $P_{\alpha}$ maps numbers $0..|\alpha|-1$, in the numbering system generated by multi-index $\alpha^{\star}$ 
to numbers with inverted representation generated by multi-index $\alpha$.

Transposition $P_{\alpha}$ generates a linear transform in ${\sf R}^N$, which is denoted by $S_{\alpha}$. 
It is fully defined by $S_{\alpha} e_{n,N} = e_{P_{\alpha}n, N}$ for $0\le n <N$. 

From the definition of $L^{mn}_n$ it follows that if $\alpha=(m,n)$ is a pair, then $S_{\alpha} = L^{mn}_n$. 

\begin{lemma} \label{inverse_permut}
Let's $\alpha$ --- multi-index and $N = |\alpha|$. Let's $M>0$ and $\beta = (M, \alpha)$ --- multi-index extended with $M$. Then
$$
S_{(M, \alpha)} = (I_M \otimes S_{\alpha}) L^{NM}_N, \qquad S_{(\alpha, M)} = L^{NM}_M (I_M \otimes S_{\alpha}).
$$
\end{lemma}

\begin{proof}
Let's $0\le m<M$ and $0\le n<N$. To prove the first equation, introduce $\beta=(M, \alpha)$. Then:
\begin{equation*}
\begin{split}
mN + P_{\alpha}n = (m, (n_{\alpha^{\star}})^{\star})^{(M, \alpha)} = [(n_{\alpha^{\star}},m)^{\star}]^{\beta} = \\ 
[((nM+m)_{\beta^{\star}})^{\star}]^\beta = P_{\beta} (m+nM).
\end{split}
\end{equation*}
Substitute the equation with the following formula:
\begin{equation*}
\begin{split}
(I_M \otimes S_{\alpha}) L^{NM}_N (e_{n, N} \otimes e_{m,M}) = \\
(I_M \otimes S_{\alpha}) (e_{m, M} \otimes e_{n, N}) = \\
e_{m, M} \otimes (S_{\alpha} e_{n, N}) = e_{m, M} \otimes e_{P_{\alpha} n, N} = e_{mN+P_{\alpha}n, MN} = \\ 
e_{P_{\beta}(m+nM)} = S_{\beta} e_{m+nM, MN} = S_{\beta} (e_{n, N} \otimes e_{m,M}),
\end{split}
\end{equation*}
which proves the first statement of the lemma \ref{inverse_permut}. 

The second statement of the lemma is proved in the same way.
\end{proof}

\section{General formula for FFT of arbitrary length with decimation in time (DIT)}
We will use the splitting rule to factorize the DFT matrix for computing using accelerator architecture on fig. \ref{fig:fft-base}.

If FFT length $N$ can be factored as $N = \prod_{k=0}^K n_k$,  
then FFT comprises only of butterflies of radices  $n_0, \ldots n_K$. 
This property is utilized for calculating FFT of length $N=2^n$. 
In the next statement, we generalize it to arbitrary $n_0, \ldots, n_K$. 

\begin{theorem} \label{FFT_iter}
Let's $\alpha = (n_K, \ldots, n_0) \in \sf N^{K+1}$ --- tuple of natural numbers and $N = \prod_{k=0}^K n_k$. For each $k=0, \ldots, K$ we define $N_k=\prod_{j=0}^k n_j$. Then
$$
{\mathcal F}_N = \left(\prod_{k=0}^K D_{k,\alpha}\right) S_{\alpha},
$$
where the matrices $D_{k,\alpha}$ are indexed right-to-left, 
$$
D_{k,\alpha} = A_k^{-1} (I_{N/n_k} \otimes {\mathcal F}_{n_k}) \widehat{W}_k A_k, \qquad 0\le k\le K,
$$
transposition matrices $A_k$ and diagonal matrices $\widehat{W}_k$ are defined by
$$
A_k = I_{N/N_k} \otimes L^{N_k}_{n_k}, \qquad \widehat{W}_k = I_{N/N_k} \otimes W^{N_k}_{n_k}.
$$
\end{theorem}

\begin{proof} 
The theorem is proved by induction. Let's $K=0$, then $N=n_0$. By definition, $L^{n_0}_{n_0}=I_{n_0}$ then $A_0 = I_N$. As $W^{n_0}_{n_0} = I_{n_0}$ then $\widehat{W}_0 = I_N$. 
After substitution, the theorem statement becomes identity: ${\mathcal F}_N={\mathcal F}_{n_0}$.

Then we prove the induction step. Let's assume the theorem is proven for $K-1\ge 0$. Then, prove for $K$. We apply the corollary \ref{FFT_split_IF} with $m=n_K$, $k=N_{K-1}$, $n=N$:
$$
{\mathcal F}_N = L^N_{N_{K-1}}(I_{N_{K-1}} \otimes {\mathcal F}_{n_K}) W^N_{n_K} L^N_{n_K} (I_{n_K} \otimes {\mathcal F}_{N_{K-1}}) L^N_{N_{K-1}}.
$$
By definition, $A_K = L^N_{n_K}$ then by lemma \ref{aux_LW}: $A_K^{-1} = L^N_{N_{K-1}}$. Then
$$
{\mathcal F}_N = D_{k, \alpha} (I_{n_K} \otimes {\mathcal F}_{N_{K-1}}) L^N_{N_{K-1}}.
$$
Let's $\alpha_{K-1} = (n_{K-1}, \ldots, n_0)$. From the assumption of the induction step and Kronecker product properties, it follows that
$$
{\mathcal F}_{N_{K-1}} = \left(\prod_{k=0}^{K-1} D_{k,\alpha_{K-1}}\right) S_{\alpha_{K-1}}.
$$
Then
\begin{equation*}
\begin{split}
(I_{n_K} \otimes {\mathcal F}_{N_{K-1}}) L^N_{N_{K-1}} = \\
\left(\prod_{k=1}^{K-1} (I_{n_K} \otimes D_{k,\alpha_{K-1}}) \right) (I_{n_K} \otimes S_{\alpha_{K-1}}) L^N_{N_{K-1}}.
\end{split}
\end{equation*}
By lemma \ref{inverse_permut} 
$$
(I_{n_K} \otimes S_{\alpha_{K-1}}) L^N_{N_{K-1}} = S_{\alpha}.
$$
By Kronecker product properties for $1\le k\le K-1$
\begin{eqnarray*}
I_{n_K} \otimes (I_{N_{K-1}} \otimes L^{N_k}_{n_k}) & = & I_{N_K} \otimes L^{N_k}_{n_k} = A_k, \\
I_{n_K} \otimes (I_{N_{K-1}} \otimes L^{N_k}_{N_{k-1}}) & = & I_{N_K} \otimes L^{N_k}_{N_{k-1}} = A_k^{-1}, \\
I_{n_K} \otimes (I_{N_{K-1}} \otimes W^{N_k}_{n_k}) & = & I_{N_K} \otimes W^{N_k}_{n_k} = \widehat{W}_k.
\end{eqnarray*}
Then
$$
I_{n_K} \otimes D_{k, \alpha_{K-1}} = D_{k, \alpha}, \qquad 1\le k\le K-1.
$$
Substitution leads to the formula
$$
(I_{n_K} \otimes {\mathcal F}_{N_{K-1}}) L^N_{N_{K-1}} = \left(\prod_{k=1}^{K-1} D_{k,\alpha} \right) S_{\alpha},
$$
which proves the induction step. 
\end{proof}

In the previous theorem, all matrices have the structure $I_{N/n_k} \otimes X$, which allows mapping all the computations for a single butterfly ${\mathcal F}_{n_k}$ 
to a single iteration of the architecture template. The single stage of FFT $D_{k,\alpha}$ can be mapped to a single instruction of the architecture template.

The FFT input data index representation, complying with the algorithm from theorem \ref{FFT_iter} is generated by multi-index $\alpha^{\star}=(n_0, \ldots, n_K)$. 
The output data index representation is generated by multi-index $\alpha$. 
FFT stages $D_{k,\alpha}$ in the algorithm from the theorem \ref{FFT_iter} start from lower digits of the multi-index.

Transpositions $A_k$ from the theorem \ref{FFT_iter} can be interpreted in terms of numbering systems. 
Matrix $L^{N_k}_{n_k}=S_{(N_{k-1}, n_k)}$ transposes index $n_k$ to the end of the multi-index which is formalized by the next lemma \ref{L_permut}.

For arbitrary multi-index $\alpha = (n_K, \ldots, n_0)$ and $k=0..K$ we can define a multi-index $\bar{\alpha}_k$ by the equation
$$
\bar{\alpha}_k = (n_K, \ldots, n_{k+1}, n_{k-1}, \ldots, n_0, n_k).
$$
As $|\alpha|=|\bar{\alpha}_k|$, then both multi-indexes generate numbering systems on the same range $0..N, N=|\alpha|-1$. 
Multi-index mapping is denoted as 
$$
p = (p_K, \ldots, p_0) \rightarrow \bar{p}^k = (p_K, \ldots, p_{k+1}, p_{k-1}, \ldots, p_0, p_k).
$$ 

\begin{lemma} \label{L_permut}
Let's $A_k$ for $1\le k\le K$ --- transposition matrices from the theorem \ref{FFT_iter}. For each $k=0, \ldots, K$ we introduce the multi-index 
$$
\beta_k = (n_K, n_{K-1}, \ldots, n_{k+1}, n_0, n_1, \ldots, n_k).
$$

Then for $k=1, \ldots, K$: $\beta_k = \overline{(\beta_{k-1})}_k$ matrix $A_k$ does the following transposition $P$:
$$
n = p^{\beta_{k-1}} \quad \rightarrow \quad P n = (\bar{p}^k)^{\beta_k}, \qquad 0\le n\le |\alpha|-1.
$$
\end{lemma}

\begin{proof}
Let's $1\le k\le K$. Equation $\beta_k = \overline{(\beta_{k-1})}_k$ follows from the definition of multi-index mapping $\bar{\alpha}_k$ and the definition of multi-index $\beta_j$ with $0\le j\le K$.
By definition,
$$
A_k = I_{N/N_k} \otimes L^{N_k}_{n_k} = I_{N/N_k} \otimes S_{(N_{k-1}, n_k)}.
$$
Let's $p = (p_K, \ldots, p_k, p_0, p_1 \ldots, p_{k-1})$ --- multi-index in a numbering system generated by $\beta_{k-1}$, and $n = p^{\beta_{k-1}}$ --- multi-index value. 
Then
$$
A_k e_{n,N} = e_{m, N/N_k} \otimes S_{(n_k, N_{k-1})} e_{\ell, N_k}, 
$$
\begin{eqnarray*}
m & = & (p_K, \ldots p_{K+1})^{(n_K, \ldots, n_{k+1})}, \\
\ell & = & (p_k, p_0, \ldots p_{k-1})^{(n_k, n_0, \ldots, n_{k-1})} = (p_k, q)^{(n_k, N_{k-1})}, \\ 
q & = & (p_0, \ldots p_{k-1})^{(n_0, \ldots, n_{k-1})}.
\end{eqnarray*}
Then 
\begin{equation*}
\begin{split}
S_{(N_{k-1}, n_k)} e_{\ell, N_k} = S_{(N_{k-1}, n_k)} (e_{p_k, n_k} \otimes e_{q, N_{k-1}}) = \\
e_{q, N_{k-1}} \otimes e_{p_k, n_k} = e_{(q, p_k)^{(N_{k-1}, n_k)},N_k}, \\
(q, p_k)^{(N_{k-1}, n_k)} = (p_0, \ldots p_k)^{(n_0, \ldots, n_k)}, \\
A_k e_{n,N} = e_{m, N/N_k} \otimes e_{(p_0, \ldots p_k)^{(n_0, \ldots, n_k)},N_k} = e_{(\bar{p}^k)^{\beta_k}, N},
\end{split}
\end{equation*}
which proves the statement of the lemma \ref{L_permut}.
\end{proof}

The FFT algorithm from theorem \ref{FFT_iter} is usually called decimation in time (DIT). 
The characteristics of this algorithm are digit-reverse transposition $S_{\alpha}$ of the input vector and twiddle multiplication before the butterfly. 

\section{General formula for FFT of arbitrary length with decimation in frequency (DIF)}

Dual FFT implementation with digit-reverse transposition of the output vector is usually called decimation in frequency (DIF).
The formula for the DIF FFT is produced by transposing the FFT factorization from theorem \ref{FFT_iter}.

\begin{corollary} \label{FFT_iter_transpos}
Let's $\alpha = (n_K, \ldots, n_0) \in \sf N^{K+1}$ and $N = \prod_{k=0}^K n_k$. For each $k=0, \ldots, K$ we define $N_k=\prod_{j=0}^k n_j$. Then
$$
{\mathcal F}_N = S_{\alpha}^{-1} \left(\prod_{k=0}^K \widehat{D}_{k,\alpha}\right) ,
$$
where the matrices $\widehat{D}_{k,\alpha}$ are indexed left-to-right,
$$
\widehat{D}_{k,\alpha} = A_k^{-1} \widehat{W}_k (I_{N/n_k} \otimes {\mathcal F}_{n_k}) A_k, \qquad 0\le k\le K,
$$
transposition matrices $A_k$ and diagonal matrices $\widehat{W}_k$ are defined by equations
$$
A_k = I_{N/N_k} \otimes L^{N_k}_{n_k}, \qquad \widehat{W}_k = I_{N/N_k} \otimes W^{N_k}_{n_k}.
$$
\end{corollary}

\begin{proof}
From the definition of the matrix ${\mathcal F}_N$ follows that it is symmetric: ${\mathcal F}_N=({\mathcal F}_N)^T$. 
The statement of the corollary is concluded by transposing the formula from theorem \ref{FFT_iter} and substituting equations $A_k^T = A_k^{-1}$, 
$(S_{\alpha})^T=S_{\alpha}^{-1}$ and $(W^m_n)^T=W^m_n$.
\end{proof}

The next theorem provides the formula for DIF FFT, with stage radices $m_0, \ldots, m_K$. 

\begin{theorem} \label{FFT_iter_freq}
Let's $\beta = (m_K, \ldots, m_0) \in \sf N^{K+1}$ and $N = \prod_{k=0}^K m_k$. For each $k=0, \ldots, K$ we define $M_k=\prod_{j=k}^K m_j$. Then
$$
{\mathcal F}_N = S_{\beta} \left(\prod_{k=0}^K \widetilde{D}_{k,\beta}\right),
$$
where the matrices $\widetilde{D}_{k,\beta}$ are indexed right-to-left, 
$$
\widetilde{D}_{k,\beta} = B_k^{-1} \widetilde{W}_k (I_{N/m_k} \otimes {\mathcal F}_{m_k}) B_k, \qquad 0\le k\le K,
$$
transposition matrices $B_k$ and diagonal matrices $\widehat{W}_k$ are defined by equations
$$
B_k = I_{N/M_k} \otimes L^{M_k}_{m_k}, \qquad \widetilde{W}_k = I_{N/M_k} \otimes W^{M_k}_{m_k}.
$$
\end{theorem}

\begin{proof}
We define $\alpha=\beta^{\star} = (m_0, \ldots, m_K)$, so $m_k=n_{K-k}$ for $0\le k\le K$. 
Then $S_{\alpha}^{-1} = S_{\alpha^{\star}}=S_{\beta}$ and using the corollary \ref{FFT_iter_transpos}: 
$$
N_k = \prod_{j=0}^k m_{K-j} = \prod_{i=K-k}^K m_i = M_{K-k}, \qquad 0\le k\le K.
$$
Then $\widehat{W_k} = \widetilde{W}_{K-k}$, $A_k = B_{K-k}$ and consequently $\widehat{D}_{k, \alpha} = \widetilde{D}_{K-k, \beta}$. 
\end{proof}
 
The FFT input data index representation, complying with the algorithm from the theorem \ref{FFT_iter_freq} is generated by the multi-index $\beta^{\star}=(m_0, \ldots, m_K)$. 
The output data index representation is generated by the multi-index $\beta$. 
FFT stages $\widetilde{D}_{k,\beta}$ in the algorithm from the theorem \ref{FFT_iter_freq} start from higher digits of the multi-index.

The FFT implementation from the theorem \ref{FFT_iter_freq} calculates twiddle multiplications with matrix $W^n_k$ after the butterfly in each stage $\widetilde{D}_{k,\beta}$. 
For better accuracy, it is desirable to multiply before the butterfly~\cite{4626107}, as the internal number representation in PU has a wider mantissa. 
Moreover, this order of computations allows unifying architecture template fig.~\ref{fig:fft-base} between DIT and DIF FFT.
The next theorem provides a modification of the theorem \ref{FFT_iter_freq}, with multiplication before the butterfly.

Let's $k, m, n \in \sf N$. We define a diagonal matrix $V_{k,m,n}$ with size $N=kmn$ by equations
$$
V_{k,m,n} (e_{i,k} \otimes e_{j,m} \otimes e_{\ell, n}) = \omega_{kmn}^{i(\ell m+j)} (e_{i,k} \otimes e_{j,m} \otimes e_{\ell, n})
$$
for $0\le i<k$, $0\le j<m$, $0\le \ell<n$.

\begin{theorem} \label{FFT_iter_freqW}
Let's $\beta = (m_K, \ldots, m_0) \in \sf N$ and $N = \prod_{k=0}^K m_k$. For each $k=0, \ldots, K$ we define $M_k=\prod_{j=k}^K m_j$. Then
$$
{\mathcal F}_N = S_{\beta} \left(\prod_{k=0}^K E_{k,\beta}\right),
$$
where the matrices $E_{k,\beta}$ are indexed right-to-left, 
$$
E_{k,\beta} = B_k^{-1} (I_{N/m_k} \otimes {\mathcal F}_{m_k}) \widetilde{X}_k B_k, \qquad 0\le k\le K,
$$
transposition matrices $B_k$ and diagonal matrices $\widehat{X}_k$ are defined by equations for $0\le k\le K$
\begin{equation*}
\begin{split}
B_k = I_{N/M_k} \otimes L^{M_k}_{m_k}, \quad \widetilde{X}_k = I_{N/M_k} \otimes V_{m_k, m_{k+1}, M_{k+2}}, \\
m_{K+1}=M_{K+1}=M_{K+2}=1.
\end{split}
\end{equation*}
\end{theorem}

\begin{proof}
In the formula in the theorem \ref{FFT_iter_freq} between neighbor butterflies $(I_{N/m_{k+1}} \otimes {\mathcal F}_{m_{k+1}})$ and $(I_{N/m_k} \otimes {\mathcal F}_{m_k})$ for $0\le k\le K{-}1$ appears the term
\begin{equation*}
\begin{split}
H_k = B_{k+1} B_k^{-1} \widetilde{W}_k = \\
(I_{N/M_{k+1}} \otimes L^{M_{k+1}}_{m_{k+1}}) (I_{N/M_k} \otimes L^{M_k}_{M_{k+1}}) (I_{N/M_k} \otimes W^{M_k}_{m_k}).
\end{split}
\end{equation*}
We want to prove that $H_k = \widetilde{X}_k B_{k+1} B_k^{-1}$, then the statement of the theorem results from the rearrangement of the terms in the formula for ${\mathcal F}_N$. 
For $k=K$ stage $\widetilde{D}_{K, \beta}$ does not contain twiddle multiplication, as $M_K=m_K$ and $W^{M_K}_{m_K} = I$ --- identity matrix. 

Let's $0\le k\le K-1$. From the properties of the Kronecker product and from lemma \ref{aux_LW}
\begin{eqnarray*}
B_k^{-1} \widetilde{W}_k & = & I_{N/M_k} \otimes (L^{M_k}_{M_{k+1}}W^{M_k}_{m_k}) \\
& = & I_{N/M_k} \otimes (W^{M_k}_{M_{k+1}} L^{M_k}_{M_{k+1}}) \\
& = & (I_{N/M_k} \otimes W^{M_k}_{M_{k+1}}) (I_{N/M_k} \otimes L^{M_k}_{M_{k+1}}) \\
& = & (I_{N/M_k} \otimes W^{M_k}_{M_{k+1}}) B_k^{-1}.
\end{eqnarray*}
The common multiplier $I_{N/M_k}$ in the Kronecker product can be factored out, so it is enough to prove that
$$
(I_{m_k} \otimes L^{M_{k+1}}_{m_{k+1}}) W^{M_k}_{M_{k+1}} = V_{m_{k+1}, m_k, M_{k-1}} (I_{m_k} \otimes L^{M_{k+1}}_{m_{k+1}}).
$$
Let's $0\le N_{k+2} < M_{k+2}$, $0\le n_k < m_k$, $0\le n_{k+1} < m_{k+1}$.
Then $0\le n_{k+1} M_{k+2} + N_{k+2} <M_{k+1}$, by substitution to the left term of the equation, it is transformed to identity
\begin{equation*}
\begin{split}
(I_{m_k} \otimes L^{M_{k+1}}_{m_{k+1}}) W^{M_k}_{M_{k+1}} \times \\
(e_{n_k,m_k} \otimes (e_{n_{k+1},m_{k+1}} \otimes e_{N_{k+2},M_{k+2}})) = \\
\omega_{M_k}^{n_k(n_{k+1} M_{k+2} + N_{k+2})} (e_{n_k,m_k} \otimes e_{N_{k+2},M_{k+2}} \otimes e_{n_{k+1},m_{k+1}})  = \\
V_{m_k, m_{k+1}, M_{k+2}} (e_{n_k,m_k} \otimes e_{N_{k+2},M_{k+2}} \otimes e_{n_{k+1},m_{k+1}}) = \\
V_{m_k, m_{k+1}, M_{k+2}} (I_{m_k} \otimes L^{M_{k+1}}_{m_{k+1}}) \times \\
(e_{n_k,m_k} \otimes (e_{n_{k+1},m_{k+1}} \otimes e_{N_{k+2},M_{k+2}})),
\end{split}
\end{equation*}
which proves the statement of the lemma. 
\end{proof}

\section{Conclusion}
The paper proposes a general approach for design and implementation of memory based vector accelerators. 
The applicability of the approach is demonstrated only for 1r1w SRAM based accelerator and DIT/DIF FFT which can be used for fast circular convolution.
The approach can be applied to other memory architectures, i.e. 1rw SRAM and other FFT such as self-sorting FFT.

\section*{Acknowledgment}
The author thanks professor \fbox{Andrey Barabanov} for his invaluable comments and advice.



\bibliographystyle{IEEEtran}
\bibliography{main}

\end{document}